%% file: root.tex
\documentclass[letterpaper, 10pt, conference]{ieeeconf}  %

\IEEEoverridecommandlockouts                              %

\usepackage{amsmath} %
\usepackage{amssymb}  %
\usepackage[noadjust]{cite}
\usepackage{graphicx}
\usepackage{pgfplots}
\pgfplotsset{compat=newest}
\usepackage{bbm}
\usepackage{mathtools}
\usepackage{dsfont}
\usepackage{aligned-overset} %

\usepackage[hidelinks]{hyperref} %

\newtheorem{thm}{Theorem}
\newtheorem{cor}{Corollary}
\newtheorem{lem}{Lemma}

\newtheorem{fact}{Fact}

\newtheorem{rem}{Remark}

\newcommand{\TTT}{\ensuremath{T_{\mathrm{TT}}}}
\newcommand{\TET}{\ensuremath{T_{\mathrm{ET}}}}
\newcommand{\expval}[1]{\ensuremath{\mathbb{E}\!\left\lbrack#1\right\rbrack}}
\newcommand{\varval}[1]{\ensuremath{\mathbb{V}\!\left\lbrack#1\right\rbrack}}
\newcommand{\Prob}{\ensuremath{\mathbb{P}}}
\newcommand*\diff{\mathop{}\!\mathrm{d}}

\title{\LARGE \bf
	Analysis of Time- versus Event-Triggered Consensus for a Single-Integrator Multi-Agent System*
\thanks{*F.\ Allgöwer thanks the German Research Foundation (DFG) for support of this work within grant AL 316/13-2 and within the German Excellence Strategy under grant EXC-2075 - 285825138; 390740016.
	F.\ Aurzada thanks the DFG for support within grant AU 370/7.
	For the cooperation until Jan. 2022, M.\ A.\ Lifshits thanks for the support within grant RFBR 20-51-12004.}}

\author{David Meister$^{1}$, Frank Aurzada$^{2}$, Mikhail A. Lifshits$^{3}$, and Frank Allgöwer$^{1}$%
\thanks{$^{1}$D.\ Meister and F.\ Allgöwer are with the University of Stuttgart, Institute for Systems Theory and Automatic Control, Stuttgart, Germany,\newline
		{\tt\small \{meister,allgower\}@ist.uni-stuttgart.de}}%
\thanks{$^{2}$F.\ Aurzada is with the Technical University of Darmstadt, Germany,
        {\tt\small aurzada@mathematik.tu-darmstadt.de}}%
\thanks{$^{3}$M.\ A.\ Lifshits is with the St. Petersburg State University, Russia, 
    	{\tt\small mikhail@lifshits.org}}%
}

\begin{document}

\maketitle

\begin{abstract}
	It is well known that the employed triggering scheme has great impact on the control performance when control loops operate under scarce communication resources.
	Various practical and simulative works have demonstrated the potential of event-triggered control to reduce communication while providing a similar performance level when compared to time-triggered control.
	For non-cooperative networked control systems, analytical performance comparisons of time- and event-triggered control support this finding under certain assumptions.
	While being well-studied in the non-cooperative setting, it remains unclear if and how the performance relationship of the triggering schemes is altered in a multi-agent system setup.
	To close this gap, in this paper, we consider a homogeneous single-integrator multi-agent consensus problem for which we compare the performance of time- and event-triggered control schemes analytically.
	Under the assumption of equal average triggering rates, we use the long-term average of the quadratic deviation from consensus as a performance measure to contrast the triggering schemes.
	Contrary to the non-cooperative setting, we prove that event-triggered control performs worse than time-triggered control beyond a certain number of agents in this setup.
	In addition, we derive the asymptotic order of the performance measure as a function of the number of agents under both triggering schemes.
\end{abstract}

\begin{keywords}
	Agents-based systems, Networked control systems, Event-triggered control, Consensus.
\end{keywords}

\section{Introduction}

As shown by \cite{Astrom2002}, event-triggered control (ETC) schemes have the potential to outperform time-triggered control (TTC) schemes for single-integrator systems when assuming loss- and delay-free communication channels as well as equal average triggering rates.
While the former only initiate communication when a designed triggering condition is fulfilled, the latter establish communication intervals of fixed length.
The reduction of ``unnecessary'' communication may appear as a valid argument in favor of ETC also being beneficial for imperfect communication channels with limited bandwidth.
The suggestion to use ETC to reduce the shared medium utilization carried over from the field of networked control systems (NCS), e.g., \cite{Henningsson2008,Heemels2008}, to the field of multi-agent system (MAS), see for example \cite{Seyboth2013,Nowzari2019}.
In order to differentiate clearly between NCS that are only coupled through their use of a shared communication medium and MAS in which the agents additionally cooperate on a common goal, we refer to the former as non-cooperative NCS throughout this paper.

For the non-cooperative single-integrator NCS case, \cite{Rabi2009} extended the results from \cite{Astrom2002} to incorporate also network effects such as packet loss in the comparison.
They point out that ETC can perform worse than TTC above a certain packet loss probability.
In \cite{Blind2011} and \cite{Blind2011a}, transmission delays are included into the comparison and the packet loss probability is determined based on the medium access protocol.
At last, \cite{Blind2013} provides a performance comparison of TTC and ETC schemes for single-integrator systems considering various medium access protocols.
They demonstrate the impact of the network load on the performance of the single-integrator NCS depending on the triggering scheme and medium access protocol.
Thereby, they establish the importance of taking the properties of the communication network into consideration when designing triggering schemes for NCS.
Analyzing more general NCS and their behavior under (periodic) ETC schemes is still an active field of research, e.g., \cite{Gleizer2020, Postoyan2022}.

Although some fundamental considerations have shown that TTC can sometimes outperform ETC if network effects are taken into account, ETC is still very popular for NCS.
This also led to various ETC approaches for MAS while there exists no work on the fundamental characteristics of TTC compared to ETC in this case.
As pointed out by \cite{Nowzari2019}, the event-triggered consensus literature is still missing performance analyses that quantify the benefit of ETC over TTC schemes.
This work aims to close this gap in order to understand whether qualitative results are the same for MAS as in the non-cooperative NCS case, or whether new effects might arise.
With the discussed works in mind, we provide a first theoretical evaluation of ETC and TTC performance by analyzing a simple MAS problem.
Our main contribution is the finding that, for this particular setup, ETC is not always superior to TTC even without considering packet loss or transmission delays.
The performance relationship turns out to depend on the number of participating agents.
Moreover, we provide the asymptotic order of the performance measure for ETC and TTC as a function of the number of agents.
This gives further insights into the relationship between ETC and TTC for MAS, considering a particular setup.

Our paper is structured as follows: 
In Section~\ref{sec:setup}, we introduce the setup and formulate the considered problem. 
After that, we present our theoretical results in Section~\ref{sec:analysis}, while we demonstrate our findings in a numerical simulation in Section~\ref{sec:sim}.
We conclude this work in Section~\ref{sec:conclusion}.

\section{Setup and Problem Formulation}\label{sec:setup}

We consider an all-to-all communication graph with $N$ nodes representing the agents.
Therefore, every agent is able to communicate directly with all other agents.
Moreover, each agent acts as a single-integrator perturbed by noise
\begin{equation}\label{eq:agent}
	\mathrm{d}x_i = u_i \mathrm{d}t + \mathrm{d}v_i
\end{equation}
starting in consensus, i.e., initial states $x_i(0)=0$ for all $i\in \lbrace 1,\dots,N\rbrace$, and with $v_i(t)\in \mathbb{R}$ referring to a standard Brownian motion and $u_i(t)\in\mathbb{R}$ to the control input.

Furthermore, we assume that the agents can continuously monitor their own state but trigger discrete transmission events.
As laid out in the introduction, we intend to compare TTC and ETC schemes for triggering transmissions.
For that purpose, we consider the cost functional
\begin{equation}\label{eq:cost}
	J \coloneqq \limsup_{M\to \infty} \frac{1}{M} \int_{0}^{M} \expval{\frac{1}{2}\sum_{i,j=1}^{N}(x_i(t)-x_j(t))^2} \diff t,
\end{equation}
as a performance measure, where we abbreviated $\sum_{i=1}^{N}\sum_{j=1}^{N}$ as $\sum_{i,j=1}^{N}$.
It quantifies the average deviation from consensus with a quadratic cost term.
\begin{rem}
	The cost can also be expressed as 
	\begin{equation*}
		J = \limsup_{M\to \infty} \frac{1}{M} \int_{0}^{M} \expval{x^\top Lx} \diff t,
	\end{equation*} where $x=\lbrack x_1(t), \dots, x_N(t)\rbrack^\top$ and $L$ is the Laplacian of the all-to-all communication graph.
	The quadratic term $x^\top Lx$ is a typical measure for the deviation of an MAS from consensus and, for example, also often used as a Lyapunov function, see, e.g., \cite{Dimarogonas2009}.
\end{rem}
\begin{rem}
	We do not incorporate the triggering rate in the cost since we intend to compare TTC and ETC given equal average triggering rates, cf.\ Section~\ref{sec:ET} and, e.g., \cite{Antunes2020}.
\end{rem}

The agents are controlled with an impulsive control input
\begin{equation}\label{eq:input}
	u_i(t) = \sum_{j\in\mathcal{N}_i} \sum_{k\in\mathbb{N}} \delta(t-t^j_{k})(x_j(t^j_{k})-x_i(t^j_{k})),
\end{equation}
where $\mathcal{N}_i = \lbrace1,\dots,N\rbrace\backslash\lbrace i\rbrace$ is the set of neighbors of agent $i$, $\delta(\cdot)$ refers to the Dirac delta function and $t^j_{k}$ denotes the transmission time of packet $k$ from agent $j$.
Given the performance measure \eqref{eq:cost}, this is the optimal control input given any sampling scheme since there is no cost induced by the control input.
Thus, the system is reset to consensus by transmitting an agent's state to all other agents.
Between those transmission events, the systems behave according to standard Brownian motions.

We consider two different ways of notation regarding the series of transmission events in this work:
On the one hand, we have already introduced the series of triggering time instances $(t^j_k)_{k\in\mathbb{N}}$ for agent $j\in\lbrace1,\dots,N\rbrace$.
On the other hand, we will also refer to the event series of the complete MAS with the notation $(t_k)_{k\in\mathbb{N}}$.
Naturally, one obtains the sequence $(t_k)_{k\in\mathbb{N}}$ by ordering the event series $(t^j_k)_{k\in\mathbb{N}}$ for all agents $j\in\lbrace1,\dots,N\rbrace$ in an increasing fashion.

As a final point, note that we study the same setup as in \cite{Astrom2002} except for the fact that we consider a cooperative control goal.
This will allow us to contrast the results later on.

\section{Main Results}\label{sec:analysis}

In this section, the two triggering schemes are introduced and the related cost according to \eqref{eq:cost} is derived and compared.

\subsection{Preliminaries}\label{sec:pre}
First, we prove some facts about the considered problem that will turn out to be useful for the following analysis of the TTC and ETC schemes.
Similar to \cite{Rabi2009}, we find:
\begin{fact}\label{fact:0T}
	If the sequence of inter-event times is independent and identically distributed, it suffices to evaluate the cost over the first sampling interval $\lbrack0,T\rbrack = \lbrack0,t_1\rbrack$:
	\begin{equation*}
		J = \frac{\expval{\frac{1}{2}\sum_{i,j=1}^{N} \int_{0}^{T} \left(x_i(t)-x_j(t)\right)^2\diff t}}{\expval{T}},
	\end{equation*}
	where $T$ is determined by the respective triggering scheme introduced in Sections~\ref{sec:TT} and \ref{sec:ET}.
\end{fact}
\begin{proof}
	Can be found in the appendix.
\end{proof}

Denoting $Q \coloneqq \expval{\frac{1}{2}\sum_{i,j=1}^{N} \int_{0}^{T} \left(x_i(t)-x_j(t)\right)^2\diff t}$, we are able to rewrite the numerator of the cost as follows.
\begin{fact}\label{fact:Q}
	Let $T$ be a symmetric stopping time, i.e., if one replaces $v_i$ by $-v_i$ for any $i\in\lbrace1,\dots,N\rbrace$ the value of $T$ does not change, as well as independent of the direction, i.e., $T$ does not change if $v_i$ is interchanged with $v_j$ for any $i,j\in\lbrace1,\dots,N\rbrace$.
	Then, given Fact~\ref{fact:0T}, we can establish
	\begin{equation*}
		Q = N(N-1) \expval{\int_{0}^{T} v_1(t)^2 \diff t}.
	\end{equation*}
\end{fact}
\vspace*{\belowdisplayskip}
\begin{proof}
	We start with the expression
	\begin{align*}
		Q ={} &\expval{ \frac{1}{2}\sum_{i,j=1}^{N} \int_0^{T} (v_i(t)-v_j(t))^2 \diff t}
		\\
		={} &\expval{ \frac{1}{2}\sum_{\substack{i,j=1:\\i\neq j}}^{N} \int_0^{T} (v_i(t)^2 -2 v_i(t) v_j(t) +v_j(t)^2) \diff t}.
	\end{align*}
	By assumption, the stopping time $T$ is symmetric.
	Observe that the distribution of the random variable $\int_0^{T} v_i(t) v_j(t)\diff t$ is symmetric as well since replacing $v_i$ by $-v_i$ only changes the sign of the integrand. 
	Therefore, the expectation of the mixed term is zero for any $i\neq j$. 
	This shows
	\begin{align*}
		Q ={} &\expval{ \int_0^{T} \sum_{1\leq i < j \leq N}  (v_i(t)^2+v_j(t)^2) \diff t}
		\\
		={}	&\expval{ \int_0^{T} \left( \sum_{i=1}^{N}  v_i(t)^2 (N-i) + \sum_{j=1}^N v_j(t)^2 (j-1) \right) \diff t}
		\\
		={} &\expval{ \int_0^{T} \sum_{i=1}^{N}  v_i(t)^2 (N-i+i-1) \diff t} 
		\\
		={} &(N-1) N \expval{ \int_0^{T}   v_1(t)^2  \diff t},
	\end{align*}
	using that $T$ is independent of the direction.
\end{proof}

\subsection{Time-Triggered Control}\label{sec:TT}
For TTC, periodic transmission events with a constant inter-event time $\TTT = t_{k+1}-t_{k} = \mathrm{const.}$ for all $k\in\mathbb{N}$ are designed for the MAS.
At each event, one agent broadcasts its state to all other agents and, thereby, triggers a reset of the MAS to consensus.
In the analyzed problem setup, the resulting cost does not depend on the choice of the transmitting agent.
This is due to the fact that the cost only considers the relative state error between the agents and all agents instantaneously move to consensus as soon as any agent communicates its state.

Deploying this triggering scheme in the considered setup allows us to arrive at the following theorem.
\begin{thm}\label{thm:cost_TT}
	Suppose agents \eqref{eq:agent} are controlled by the impulsive input \eqref{eq:input} with constant inter-event times $\TTT$.
	Then, the cost \eqref{eq:cost} is given by
	\begin{equation*}
		J_{\mathrm{TT}}(\TTT) = N(N-1) \frac{\TTT}{2}.
	\end{equation*}
\end{thm}
\vspace*{\belowdisplayskip}
\begin{proof}
	Since the inter-event times $\TTT$ are identical and constant, it suffices to analyze the interval between two transmissions and Facts~\ref{fact:0T} and \ref{fact:Q} hold.
	Thus, we can write \eqref{eq:cost} as $J_{\mathrm{TT}}(\TTT) = Q_\mathrm{TT}(\TTT)/\TTT$ with
	\begin{align*}
		Q_\mathrm{TT}(\TTT) &= N(N-1) \int_{0}^{\TTT} \expval{v_1(t)^2} \diff t \\ 
		&= N(N-1) \int_{0}^{\TTT} t \diff t = N(N-1) \frac{\TTT^2}{2}.
	\end{align*}
\end{proof}

\begin{rem}
	The obtained result for the cost in the TTC case is the same as in \cite{Astrom2002} but scaled by the number of agent pairs $N(N-1)$.
	This is also similar to the results of \cite{Blind2011} and related papers where non-cooperative NCS are analyzed and the cost therefore scales with the number of network participants $N$.
\end{rem}

\subsection{Event-Triggered Control}\label{sec:ET}
In ETC, a triggering condition is continuously evaluated to determine when an agent's state is to be transmitted instead of fixing the inter-event time.
In the literature, it is often argued that this lowers the communication rate while maintaining the same performance level, see, e.g., \cite{Seyboth2013}.
In order to achieve a desirable performance, the triggering condition is designed such that it indicates communication necessity in the considered setup.
For this work, we define
\begin{equation}\label{eq:ET_cond}
	\lvert x_i(t)-x_i(t_{\hat{k}}) \rvert \geq \Delta,
\end{equation}
where $\hat{k} = \max \left\lbrace k\in\mathbb{N} \mid t_{k} \leq t \right\rbrace$ and $\Delta>0$, as the triggering condition.
Comparing the deviation of the current state $x_i(t)$ of agent $i$ from its state at the last event $x_i(t_{\hat{k}})$ to a threshold $\Delta$ is quite common in a distributed setup since each agent can evaluate this condition with its local information, see, e.g., \cite{Dimarogonas2009}.
We choose the same threshold $\Delta$ for all agents since the contribution of each agent's state to the cost is equal.
Note that our triggering condition is analogous to the one used in \cite{Astrom2002}, \cite{Rabi2009}, \cite{Blind2013}.

As in the TTC case, we can again analyze the cost over the first sampling interval by utilizing Facts~\ref{fact:0T} and \ref{fact:Q}.
In the ETC case, the inter-event time for that interval is a stochastic variable and can be defined as a stopping time $\TET = \inf\lbrace t>0 \mid \exists i\in\lbrace1,\dots,N\rbrace: \lvert x_i(t)\rvert = \Delta \rbrace$.
The probabilistic nature of the inter-event times combined with the coupling between the agents do not allow to derive a closed form solution for the cost $J_\mathrm{ET}$ according to \eqref{eq:cost}.
However, we are also not primarily interested in the explicit cost, but rather in the relationship to the TTC cost from Theorem~\ref{thm:cost_TT}.
For a fair comparison, we choose $\TTT=\expval{\TET}$ which results in the same average number of transmissions for both triggering schemes.

By establishing the following fact, we can concentrate on $\Delta=1$ for the derivations to come.
\begin{fact}\label{fact:scaling}
	We can show that the following scaling relationships hold true:
	\begin{align*}
		Q_\mathrm{ET}(\Delta=\delta) &= \delta^4 Q_\mathrm{ET}(\Delta=1), \\
		\expval{\TET\mid\Delta=\delta} &= \delta^2 \expval{\TET\mid\Delta=1}, \\
		\varval{\TET\mid\Delta=\delta} &= \delta^4 \varval{\TET\mid\Delta=1},
	\end{align*}
	where $Q_\mathrm{ET}$ refers to $Q$ for the event-triggered case, i.e., with upper integral limit $\TET$, and $\mathbb{V}[\cdot]$ denotes the variance.
\end{fact}
\begin{proof}
	Can be found in the appendix.
\end{proof}

Without loss of generality, we will thus concentrate on $\Delta=1$ for the remainder of this section.
Before deriving the main result of this section, we need to analyze the asymptotic order of the moments of $\TET$.

\begin{lem}\label{lem:asymp_order}
	For $\Delta=1$, we have
	\begin{align} %
		\expval{\TET} &\sim \frac{1}{2\ln N}, \label{eq:expTET} \\
		\expval{\TET^2} &\sim \frac{1}{(2\ln N)^2}, \\
		\varval{\TET} &\sim \frac{\pi^2/24}{ (\ln N)^4}, \label{eq:varTET}
	\end{align}
	where $a_N \sim b_N$ means that $\lim_{N\to \infty}a_N/b_N = 1$ for arbitrary series $(a_N)_{N\in\mathbb{N}}$, $(b_N)_{N\in\mathbb{N}}$.
\end{lem}
\begin{proof}
	Let $T_j \coloneqq \inf\lbrace t>0 : \lvert x_j(t)\rvert = 1 \rbrace$ for all $j\in\lbrace1,\dots,N\rbrace$ and, thus, $\TET = \inf_{1\leq j\leq N} T_j$.
	Using the tail behavior derived from \cite{Moerters2010}, Theorem~7.45,
	\begin{align*}
		\Prob(T_j\le w) &= \Prob( \sup_{0\leq t\leq w} |v_j(t)| \geq 1 )
		\\
		&= \Prob( \sup_{0\leq t\leq 1} |v_j(t)| \geq w^{-1/2} )
		\\
		\overset{w\to 0}&{\sim} \frac{\kappa}{w^{-1/2}}\, \exp( -w^{-1}/2), %
	\end{align*}
	for $\kappa=\sqrt{2/\pi}$, and the independence of the exit times $T_j$, one can derive the limit theorem
	\begin{equation} \label{eq:limittheorem}
		2(\ln N)^2 \left({\TET}- a_N\right) \Rightarrow G,  \qquad \textrm{as } N\to\infty,
	\end{equation}
	with 
	\begin{equation*}
		a_N\coloneqq\frac{1}{2\ln N}-\frac{\ln \frac{\kappa}{(2\ln N)^{1/2}}}{2(\ln N)^2},
	\end{equation*}
	and where $\Rightarrow$ stands for convergence in distribution.
	Moreover, $G$ is a Gumbel-distributed random variable,
	\begin{equation*}
		\Prob( G\ge r) = \exp( - \exp(r)).
	\end{equation*}
	Equation \eqref{eq:limittheorem} can be derived from \cite{Galambos1978}, Theorem~2.1.6. 
	A direct proof is given here: Indeed, for any $r\in\mathbb{R}$, we have
	\begingroup		%
	\allowdisplaybreaks
	\begin{align*}
		&\Prob( 2(\ln N)^2 \left({\TET}- a_N\right)\geq r)
		\\
		={} &\Prob( {\TET} \geq \frac{r}{2(\ln N)^2} + a_N)
		\\
		={}& \Prob( \forall j=1,\ldots, N : T_j \geq \frac{r}{2(\ln N)^2} + a_N)
		\\
		={}& \Prob( T_1 \geq \frac{r}{2(\ln N)^2} + a_N)^N
		\\
		={}& \left( 1 - \Prob( T_1 < \frac{r}{2(\ln N)^2} + a_N )\right)^N
		\\
		\sim{}&	\left(1 - c_N \exp\!\left( -\frac{1}{2} \left( \frac{r -\ln c_N}{2(\ln N)^2}+ \frac{1}{2\ln N} \right)^{-1} \right) \right)^N
		\\
		={}& \left(1-c_N \exp\!\left( -\ln N \left( \frac{r -\ln  \frac{\kappa}{(2\ln N)^{1/2}}}{\ln N}+ 1 \right)^{-1} \right) \right)^N 
		\\
		\sim{}& \left(1-c_N \exp\!\left( -\ln N \left(1- \frac{r -\ln  \frac{\kappa}{(2\ln N)^{1/2}}}{\ln N} \right) \right) \right)^N
		\\
		={}& \left(1 - c_N \frac{1}{N} \exp( r  -\ln c_N) \right)^N
		\\
		={}& \left(1 - \frac{1}{N} \exp( r ) \right)^N  \sim  e^{-e^r},
	\end{align*}
	\endgroup
	as required and with $c_N=\kappa/(2\ln N)^{1/2}$.
	
	The limit theorem \eqref{eq:limittheorem} is accompanied by the convergence of the first and second moment.
	The proof for this is omitted due to space limitations.
	It builds upon Lebesgue's dominated convergence theorem where we need to show that $\Prob( 2(\ln N)^2 \left({\TET}- a_N\right)\geq r)$ and $2r\Prob( 2(\ln N)^2 \left({\TET}- a_N\right)\geq r)$ are upper bounded by integrable functions.
	
	Taking expectations in \eqref{eq:limittheorem} gives
	\begin{equation*}
		2 (\ln N)^2 ( \expval{\TET} - a_N) \to \expval{G}.
	\end{equation*}
	This shows
	\begin{align}
		\expval{\TET} &= a_N + \frac{\expval{G}}{2 (\ln N)^2} (1+o(1)) \label{eq:expTET_1}\\
		&= \frac{1}{2\ln N} + \mathcal{O}\!\left( \frac{\ln \ln N}{(\ln N)^2}\right). \nonumber
	\end{align}
	Similarly, taking second moments in \eqref{eq:limittheorem} gives 
	\begin{equation*}
		4 (\ln N)^4 \expval{( \TET - a_N)^2} \to \expval{G^2}.
	\end{equation*}
	This shows
	\begin{equation*}
		\expval{\TET^2} - 2 a_N \expval{\TET} + a_N^2 = \frac{\expval{G^2}}{4(\ln N)^4} (1+o(1)),
	\end{equation*}
	which, together with \eqref{eq:expTET_1}, yields
	\begin{align*}
		\expval{\TET^2} ={}& a_N^2 + 2 a_N \frac{\expval{G}}{2 (\ln N)^2} (1+o(1)) \\
		&+\frac{\expval{G^2}}{4(\ln N)^4} (1+o(1)) \\
		={}& \frac{1}{4 (\ln N)^2} + \mathcal{O}\!\left( \frac{1}{(\ln N)^3}\right).
	\end{align*}
	Finally, the limit theorem can be re-written as
	\begin{equation*}
		2 (\ln N)^2 ( \TET - \expval{\TET}) + 2(\ln N)^2 ( \expval{\TET} - a_N) \Rightarrow G.
	\end{equation*}
	Squaring, taking expectations, and dividing by $4(\ln N)^4$ gives 
	\begin{equation*}
		\expval{( \TET - \expval{\TET})^2}  +  ( \expval{\TET} - a_N)^2 = \frac{\expval{G^2}}{4(\ln N)^4} (1+o(1)).
	\end{equation*}
	This implies
	\begin{align*}
		\varval{\TET} &= \frac{\expval{G^2}}{4(\ln N)^4} (1+o(1)) -   ( \expval{\TET} - a_N)^2 \\
		&= \frac{\expval{G^2}}{4(\ln N)^4} (1+o(1)) -   \frac{\expval{G}^2}{4 (\ln N)^4} (1+o(1)) \\
		&= \frac{\varval{G}}{4 (\ln N)^4} (1+o(1)),
	\end{align*}
	which proves \eqref{eq:varTET} because $\varval{G}=\pi^2/6$.
\end{proof}

Note that we have shown a dependence of the moments of $\TET$ on the number of agents.
This is an important difference to the non-cooperative NCS case, e.g., in \cite{Astrom2002}, \cite{Rabi2009}, \cite{Blind2013}, and caused by the coupling between the agents in the cost.
Building upon Lemma~\ref{lem:asymp_order}, we arrive at the following theorem.

\begin{thm}\label{thm:cost_ET}
	Suppose agents \eqref{eq:agent} are controlled by the impulsive input \eqref{eq:input} with inter-event times $\TET = \inf\lbrace t>0 \mid \exists i\in\lbrace1,\dots,N\rbrace: \lvert x_i(t)\rvert = \Delta \rbrace$.
	Then, there exists an $N_0$ such that for all $N\geq N_0$, we have
	\begin{equation*}
		J_\mathrm{ET} > J_\mathrm{TT}(\expval{\TET}),
	\end{equation*}
	where we denote by $J_\mathrm{TT}(\expval{\TET})$ the cost under constant inter-event times $\TTT=\expval{\TET}$.
\end{thm}
\begin{proof}
	We can once more concentrate on $\Delta=1$ due to Fact~\ref{fact:scaling}.
	As before, let $T_j \coloneqq \inf\lbrace t>0 : \lvert x_j(t)\rvert = 1 \rbrace$ for all $j\in\lbrace1,\dots,N\rbrace$ and, thus, $\TET = \inf_{1\leq j\leq N} T_j$.
	Moreover, let $\tau \coloneqq \inf_{2\le j\le N} T_j \ge {\TET}$.
	
	The key estimate is
	\begin{equation} \label{e10}
		\int_0^{\TET} v_1(t)^2 \diff t 
		\ge   \int_0^\tau v_1(t)^2 \diff t - \int_0^\tau v_1(t)^2 \diff t \ \mathds{1}_{\tau\not ={\TET}},  
	\end{equation}
	where $\mathds{1}$ denotes the indicator function.
	
	Let us evaluate the expectations. 
	By independence of $\tau$ and $v_1$, we have
	\begin{align}
		\expval{\int_0^\tau v_1(t)^2 \diff t} &=  \int_0^\infty \expval{ \mathds{1}_{t\leq \tau} v_1(t)^2} \diff t \nonumber \\
		&= \int_0^\infty \expval{ \mathds{1}_{t\leq \tau}} \expval{ v_1(t)^2} \diff t \nonumber \\
		&=\expval{\int_0^\tau t \diff t} = \frac{\expval{\tau^2}}{2} > \frac{\expval{\TET^2}}{2}, \label{e11}
	\end{align}
	while since $\tau\leq T_2$ and using the Cauchy-Schwarz inequality
	\begingroup		%
	\allowdisplaybreaks
	\begin{align}
		&\expval{ \int_0^\tau v_1(t)^2 \diff t \cdot \mathds{1}_{\tau\not ={\TET}} } \nonumber\\
		\le{}& \expval{ \int_0^{T_2} v_1(t)^2 \diff t \cdot \mathds{1}_{\tau\not ={\TET}} }	\nonumber\\
		\le{}& \expval{ \left(\int_0^{T_2} v_1(t)^2 \diff t \right)^2 }^{1/2} \cdot \expval{ \mathds{1}^2_{\tau\not ={\TET}} }^{1/2} \nonumber \\
		={}& C \cdot \Prob (\tau\not ={\TET})^{1/2} = C \, N^{-1/2}, \label{e12}
	\end{align}
	\endgroup
	where $C=\mathbb{E}[ (\int_0^{T_2} v_1(t)^2 \diff t)^2 ]^{1/2}$ does not depend on the dimension $N$. 
	The last step holds because $\tau\neq\TET$ if and only if the process $v_1$ is the first to exit $[-1,1]$. 
	By symmetry, this has probability equal to $1/N$.
	Putting \eqref{e10}, \eqref{e11}, and \eqref{e12} together, one obtains
	\begin{align}
		\expval{ \int_0^{\TET} v_1(t)^2 \diff t } &> \frac{\expval{\TET^2}}{2} -  C \, N^{-1/2} \nonumber\\
		&= \frac{\expval{\TET}^2}{2} + \frac{\varval{\TET}}{2}  -  C \, N^{-1/2}. \label{e13}
	\end{align}
	Next, the definition of the limit shows that \eqref{eq:varTET} implies
	\begin{equation*}
		\frac{\varval{\TET}}{(\pi^2/24)/(\ln N)^4} > \frac{1}{2}
	\end{equation*}
	for all $N\geq N_1$. 
	Furthermore, let $N_2$ be such that $\frac{1}{4}\cdot \frac{\pi^2/24}{(\ln N)^{4}} -  C \, N^{-1/2} > 0$ for all $N\geq N_2$ and set $N_0\coloneqq\max(N_1,N_2)$. 
	Plugging the inequalities into \eqref{e13}, we see that for $N\geq N_0$
	\begin{align*}
		\frac{1}{N(N-1)}\, Q_\mathrm{ET} &=  \expval{ \int_0^{\TET} v_1(t)^2 \diff t }
		\\
		&>
		\frac{\expval{{\TET}}^2}{2} + \frac{1}{4}\cdot \frac{\pi^2/24}{(\ln N)^{4}} -  C \, N^{-1/2}
		\\
		&> \frac{\expval{{\TET}}^2}{2}
		=  \frac{1}{N(N-1)}\, Q_\mathrm{TT}(\expval{\TET}),
	\end{align*}
	where we also used Fact~\ref{fact:Q} in the first step and Theorem~\ref{thm:cost_TT} in the last step. 
	Multiplying both sides with $N(N-1)/\expval{\TET}$ gives the desired inequality.
\end{proof}

Thus, we are able to prove that ETC is not necessarily superior to TTC if we move to a setup with cooperative agents.
It is therefore not possible to simply transfer the result from \cite{Astrom2002}, that ETC schemes outperform TTC, to this very similar setup.
With the previous results, we can also derive the asymptotic order of the performance measure in the following corollary.

\begin{cor}
	The asymptotic order of the cost \eqref{eq:cost} as a function of the number of agents under both triggering schemes can be expressed as
	\begin{equation*}
		J_\mathrm{ET} \sim J_\mathrm{TT}(\expval{\TET}) \sim \frac{N(N-1)}{4\ln N}.
	\end{equation*}
\end{cor}
\vspace*{\belowdisplayskip}
\begin{proof}
	Utilizing Theorem~\ref{thm:cost_TT} and plugging in \eqref{eq:expTET} shows the relation for $J_\mathrm{TT}(\expval{\TET})$. 
	
	The lower bound for $J_\mathrm{ET}$ follows from Theorem~\ref{thm:cost_ET}. 
	For the upper bound, observe that
	\begin{align*}
		\expval{\int_0^{\TET} v_1(t)^2 \diff t} &\leq \expval{\int_0^{\tau} v_1(t)^2 \diff t} = \frac{\expval{\tau^2}}{2} \\
		&\sim \frac{1}{2(2 \ln (N-1))^2} \sim \frac{1}{2(2 \ln N)^2},
	\end{align*}
	where we used the notation from the last proof and the fact that $\tau$ has the same distribution as $\TET$ for the dimension $N-1$.
	Utilizing $J_\mathrm{ET} = N(N-1)\mathbb{E}[\int_0^{\TET} v_1(t)^2 \diff t]/\mathbb{E}[\TET]$ together with \eqref{eq:expTET} yields the desired result.
\end{proof}

We have shown that the cost for both triggering schemes increases with the same order for large numbers of agents.
Note that this does not imply that the difference between $J_\mathrm{ET}$ and $J_\mathrm{TT}(\expval{\TET})$ vanishes for large $N$.
Moreover, the cost increase is caused purely by the cooperative nature of the considered setup.
To be more precise, we do not consider network effects in this analysis which could degrade the performance level of ETC even further as shown by \cite{Rabi2009} and \cite{Blind2013}.
Together with Theorem~\ref{thm:cost_ET}, we have therefore demonstrated that ETC schemes, although they are superior to TTC in the non-cooperative setup under the assumption of delay- and loss-free communication, might lose this property in a cooperative setup.
In addition, the relationship between TTC and ETC can depend on the number of agents or network participants $N$ in the cooperative case.

\section{Simulation}\label{sec:sim}

In this section, we support our theoretical findings by simulations.
We simulate the aforementioned MAS including the impulsive control law under the two triggering strategies.
This allows us to estimate the cost ratio $J_\mathrm{ET}/J_\mathrm{TT}$ for a varying number of agents $N$ and relate the two schemes in terms of performance.
In order to facilitate a fair comparison, we choose $\TTT=\expval{\TET}$ which results in the same average number of transmissions for both triggering schemes.

Note that the cost ratio is not influenced by the choice of $\Delta$ due to the scaling property of Brownian motions described in Fact~\ref{fact:scaling}.
Plugging the results from this fact into the cost ratio leads to a cancellation of $\delta$ from the fraction.

Thus, without loss of generality, we set $\Delta=1$ for the simulations.
We simulate the MAS for $N\in\lbrace2,12,22,\dots,72\rbrace$ with $10~000$ Monte Carlo runs each.
Moreover, we use a simulation step size of $10^{-4}$s in the utilized Euler-Maruyama method.
Since it is sufficient to analyze the first sampling interval for the cost, we can end a Monte Carlo run as soon as the first event time is reached.
The resulting cost ratios are shown in Fig.~\ref{fig:cost_ratio}.
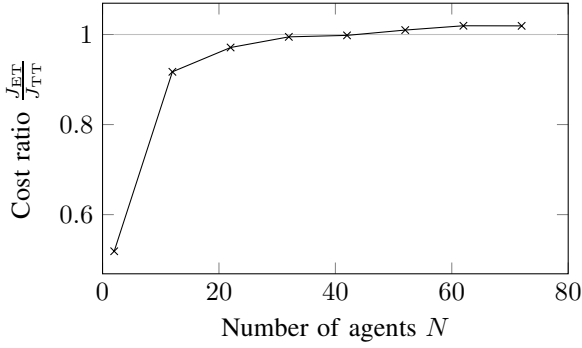
\begin{figure}
	\centering
	\input{figs/plot_cost_ratio.tex}
	\setlength{\abovecaptionskip}{0pt}
	\caption{Cost ratio of event- over time-triggered control.}
	\label{fig:cost_ratio}
\end{figure}

Although not guaranteed by our theoretical statements, we observe a clear performance advantage of the ETC scheme for low numbers of agents $N$.
Furthermore, the TTC scheme outperforms the ETC scheme for larger numbers of agents $N$.
In addition, it is interesting to note that $N_0$ from Theorem~\ref{thm:cost_ET} seems to lie in the magnitude of two-digit numbers.
This underscores that, for this particular setup, the found performance degradation of the ETC scheme can play a role for practically relevant numbers of agents $N\geq N_0$.

\section{Conclusion}\label{sec:conclusion}

In this work, we analyzed the performance of TTC and ETC in a MAS consensus setup with single-integrator agents and an all-to-all communication topology.
We proved for this particular setting that, in contrast to similar previously analyzed non-cooperative NCS setups, TTC outperforms ETC beyond a certain number of agents.
In addition, we derived the asymptotic order of the performance measure in the number of agents.
Note that the result holds without taking network effects into account.
Moreover, we demonstrated our findings in a numerical simulation.

This work highlights that ETC schemes might induce unforeseen additional effects in the MAS case when compared to non-cooperative setups.
These effects can lead to performance degradation as shown for the considered setting.
As an additional practical implication, special care must be taken when verifying ETC schemes for MAS since their superiority over TTC might critically depend on the number of agents.
A thoughtful consideration of the impact of the number of agents on such performance comparisons is of utmost importance when creating new triggering schemes and transferring knowledge from the non-cooperative NCS to the MAS domain.

In future work, we aim to include statements for more general topologies.
Moreover, we plan to incorporate network effects in the analysis.
Previous work in the non-cooperative NCS field has shown the importance of such considerations.
Since a main argument for the consideration of different triggering schemes is the operation under scarce communication resources, this is the natural next step.

\section*{APPENDIX}

\subsection{Proof of Fact~\ref{fact:0T}}
First, we compute as follows
\begingroup		%
\allowdisplaybreaks
\begin{align*}
	&\expval{\int_{0}^{M} \frac{1}{2}\sum_{i,j=1}^{N} \left(x_i(t)-x_j(t)\right)^2\diff t} \\
	={} &\frac{1}{2}\sum_{\substack{i,j=1:\\i\neq j}}^{N} \expval{\int_{0}^{M} \left(x_i(t)-x_j(t)\right)^2\diff t} \\
	={} &\frac{1}{2}\sum_{\substack{i,j=1:\\i\neq j}}^{N} \left(\expval{\sum_{k=1}^{m(M)}\int_{t_{k-1}}^{t_{k}}\left(x_i(t)-x_j(t)\right)^2\diff t}\right.\\
	&\left. + \expval{\int_{t_{m(M)}}^{M}\left(x_i(t)-x_j(t)\right)^2\diff t}\right),
\end{align*}
\endgroup
where $(m(M))_{M\in\lbrack0,\infty)}$ is a renewal process for the renewal time series $(t_k)_{k\in\mathbb{N}}$.
Since the sequence of inter-event times is independent and identically distributed, the quantities
\begin{equation*}
	y_k^{(i,j)} \coloneqq \int_{t_{k-1}}^{t_{k}}\left(x_i(t)-x_j(t)\right)^2\diff t
\end{equation*}
are i.i.d.
To see this, let $\bar{v}_i(t) = v_i(t) - v_i(t_k)$ for all $t\in\lbrack t_k,t_{k+1})$ and $i\in\lbrace1,\dots,N\rbrace$.
Note that $x_i(t) = x_i(t_k) + \bar{v}_i(t)$ for all $t\in\lbrack t_k,t_{k+1})$, $i\in\lbrace1,\dots,N\rbrace$ and $x_i(t_k)=x_j(t_k)$ for all $i,j\in\lbrace1,\dots,N\rbrace$.

By Wald's equation, we have $\mathbb{E}[\sum_{k=1}^{m(M)} y_k^{(i,j)}] = \expval{m(M)} \mathbb{E}[y_1^{(i,j)}]$.
Furthermore, the second term can be upper bounded by
\begin{equation*}
	\int_{t_{m(M)}}^{M}\left(x_i(t)-x_j(t)\right)^2\diff t \leq y_{m(M)+1}^{(i,j)}.
\end{equation*}
Dividing by $M$ and letting $M\to \infty$ shows that
\begin{align*}
	J &= \frac{1}{2}\sum_{\substack{i,j=1:\\i\neq j}}^{N} \lim_{M\to \infty} \frac{\expval{m(M)}}{M}\cdot\expval{y_1^{(i,j)}} \\
	&= \frac{1}{\expval{T}} \cdot \frac{1}{2}\sum_{i,j=1}^{N} \expval{\int_{0}^{T} \left(x_i(t)-x_j(t)\right)^2\diff t},
\end{align*}
since, by the renewal theorem, $\lim_{M\to \infty} \frac{\expval{m(M)}}{M} = \frac{1}{\expval{T}}$.

\subsection{Proof of Fact~\ref{fact:scaling}}
Let us show the first equality.
Indeed,
\begingroup		%
\allowdisplaybreaks
\begin{align*}
	&\frac{Q_\mathrm{ET}(\Delta=\delta)}{N(N-1)} \\
	={} &\expval{\int_0^{\TET} v_1(s)^2 \diff s}
	\\
	={} &\expval{\int_0^{\inf\{ t>0 \mid \exists k : \lvert v_k(\delta^2 t / \delta^2)\rvert = \delta\}} v_1(\delta^2 s/\delta^2)^2 \diff s}
	\\
	={} &\expval{\int_0^{\inf\{ t>0 \mid \exists k : \delta\lvert v_k( t / \delta^2)\rvert = \delta\}} \delta^2 v_1(s/\delta^2)^2 \diff s}
	\\
	={} &\delta^2 \expval{\int_0^{\delta^{-2} \inf\{ \delta^2 t' >0 \mid \exists k : \lvert v_k( t')\rvert = 1\}} v_1(s')^2 \delta^2 \diff s'}
	\\
	={} &\delta^4 \expval{\int_0^{\inf\{ t' >0 \mid \exists k : \lvert v_k( t')\rvert = 1\}}  v_1(s')^2 \diff s'}
	\\
	={} &\delta^4 \frac{Q_\infty(\Delta=1)}{N(N-1)}.
\end{align*}
\endgroup
Here, we utilized the scaling property of Brownian motions in the third step and linear integral substitution in the fourth step.
All other formulas are proved in a very similar fashion.

\bibliographystyle{IEEEtran}
\bibliography{IEEEabrv,references}

\end{document}

%% file: figs/plot_cost_ratio.tex
\begin{tikzpicture}
	\begin{axis}[
		xlabel=Number of agents $N$,
		ylabel=Cost ratio $\frac{J_\mathrm{ET}}{J_\mathrm{TT}}$,
		width=0.9*\linewidth,
		height=0.6*\linewidth,
		xmin=0,
		xmax=80
		]
		\addplot[mark=none, gray!50, samples=2, domain=2:78] {1};
		\addplot[mark=x,solid] table[x=noofagents,y=ratio,col sep=comma]{figs/data.csv};
	\end{axis}
\end{tikzpicture}